\documentclass[]{IEEEtran}
\usepackage[T1]{fontenc}
\usepackage[utf8]{inputenc}
\usepackage[dvips]{graphicx}
\usepackage{eepic}
\usepackage{booktabs}
\usepackage{array}
\newcolumntype{L}[1]{>{\raggedleft\let\newline\\\arraybackslash\hspace{0pt}}m{#1}}
\newcolumntype{C}[1]{>{\centering\let\newline\\\arraybackslash\hspace{0pt}}m{#1}}
\usepackage{subcaption}
\usepackage{amsmath,amssymb,amsthm,mathtools}
\usepackage{textcomp}
\usepackage{cite}
\usepackage{algorithm}
\usepackage{algpseudocode}
\usepackage[table,dvipsnames]{xcolor}
\usepackage[utf8]{inputenc}
\usepackage{cite}
\usepackage{url}
\usepackage{balance}
\usepackage{tikz}
\usetikzlibrary{arrows.meta,positioning,calc,decorations.pathreplacing}

\usepackage[pagebackref=false,colorlinks,linkcolor=black,citecolor=black]{hyperref}
\usepackage[pagebackref=false]{hyperref}

\newtheorem{theorem}{Theorem}
\newtheorem{lemma}{Lemma}

\title{$q$-ary GRAND}
\author{Reza Dianat, Majid Nasiri Khormuji,~\IEEEmembership{Senior Member,~IEEE}, and Hiroyuki Seki%
\thanks{Reza Dianat is with the Electrical Engineering Department, Persian Gulf University, Bushehr, Iran (e-mail: dianat@pgu.ac.ir). Majid Nasiri Khormuji  and Hiroyuki Seki are with 1Finity Inc., Japan (e-mail: majid.n.k@ieee.org, hseki@fujitsu.com).}}

\begin{document}
\maketitle

\begin{abstract}
We develop $q$-ary Guessing Random Additive Noise Decoding
(GRAND) for linear codes over a $q$-ary alphabet.
The decoder works on a sorted symbol-likelihood array and
uses three local child generation rules to generate symbol-deviation patterns. We then prove
that these rules induce a monotone rooted spanning tree of the full row-index
space, so best-first traversal gives maximum-likelihood (ML) decoding under
unlimited search. Reed--Solomon simulations verify ML agreement and show the finite-budget performance--complexity tradeoff.
\end{abstract}

\section{Introduction}
Guessing Random Additive Noise Decoding (GRAND) decodes by ordering noise patterns rather than codewords. It queries candidate noise patterns from most to least likely and accepts the first pattern whose removal from the received word produces a valid codeword \cite{duffy2019grand}. With exact ordering, this stopping rule gives the maximum-likelihood (ML) decision. Soft-input GRAND (SGRAND) extends the same principle to soft-output channels, including binary-input additive white Gaussian noise (AWGN), and remains ML under unlimited search \cite{sgrand}. SGRAND's exact enumeration, however, is recursive and uses dynamic data structures, which complicates parallel hardware implementation. Ordered Reliability Bits GRAND (ORBGRAND) was introduced to make the search more implementation-friendly by replacing exact reliability ordering with structured ordered-reliability approximations \cite{orbgrand}. This gives a highly effective binary-input GRAND architecture, but its ordering is approximate and its natural formulation is tied to binary ordered-reliability patterns.

For non-binary codes and higher-order modulations, the decoder receives ordered symbol likelihoods at each coordinate, so a purely bit-reliability view does not fully expose the available structure. Prior GRAND-based work has considered channels with non-binary inputs \cite{a-sol}, and several approximate GRAND-based algorithms for non-binary modulations have also been proposed \cite{sorb-fc,sorb,orbcc,orbRS, orb-hw, sog}. However, the symbol-level construction in \cite{a-sol} can generate a number of child patterns that grows with the block length, while approximate variants trade exact likelihood ordering for reduced complexity. Thus, to the best of our knowledge, existing symbol-level GRAND constructions do \emph{not} simultaneously provide bounded local branching and a proof of global likelihood-ordered coverage of the full \(q\)-ary pattern space.

We therefore develop \(q\)-ary GRAND as a symbol-level construction for length-\(n\) linear codes over a \(q\)-ary finite field alphabet $\mathbb{F}_q$. The decoder first forms a sorted \(q\times n\) likelihood array. Within each coordinate, vertical sorting ranks the \(q\) candidate symbols by likelihood; across coordinates, horizontal sorting orders the first-deviation costs. Candidate patterns are then represented by row-index vectors in \(\{0,\ldots,q-1\}^n\), where each entry specifies which ranked symbol is selected in the corresponding ordered coordinate. From any non-root pattern, \(q\)-ary GRAND generates at most three children. The first starts a new deviation, the second deepens the current deviation, and the third shifts a depth-one deviation to the next ordered coordinate. This contrasts with \cite{a-sol}, where the child-generation burden \emph{grows} with block length; here, each parent generates at most three children, independent of block length, modulation alphabet size, and code alphabet size.

Our main contribution is to show that this \emph{bounded} local generation is not merely a heuristic reduction of the search space. We prove that the three child rules induce a rooted spanning tree over the full row-index space and that every generated edge is likelihood non-increasing. Consequently, best-first traversal visits candidate patterns in \emph{global} non-increasing likelihood order, and unlimited-search \(q\)-ary GRAND returns the ML codeword. We then discuss algorithmic cost and implementation considerations, and evaluate the decoder on non-binary Reed-Solomon and binary extended Bose-Chaudhuri-Hocquenghem (eBCH) benchmarks against Wolf's syndrome-trellis ML decoder\cite{wolf}, Berlekamp-Massey (BM) decoding \cite{bla-a}, symbol-level ORBGRAND \cite{sorb}, ordered statistics decoding (OSD) \cite{foss}, and Chase decoding \cite{chase}.
\par The remainder of this letter is organized as follows. Section~\ref{sec:qary} presents the
$q$-ary GRAND search construction and decoding algorithm. Section~\ref{sec:ml} proves
ML optimality, Section~\ref{sec:wc} discusses complexity, Section~\ref{sec:num} reports
representative numerical results, and Section~\ref{sec:con} concludes the letter.

\section{$q$-ary GRAND Construction}\label{sec:qary}
Consider a length-$n$, dimension-$k$ linear code
\begin{equation}
\mathcal{C}(n,k)\subseteq \mathbb{F}_q^n,
\end{equation}
with parity-check matrix $ \mathbf{H}\in\mathbb{F}_q^{(n-k)\times n},$
where $q$ denotes the size of the code alphabet. Let
\begin{equation}
\mathcal{A}=\{u_0,\ldots,u_{M-1}\},
\end{equation}
denote the modulation constellation, where $M=|\mathcal{A}|$.
For notational simplicity, the construction and proof below are presented
for $M=q$. The method itself is not restricted to this case and also applies
when $M\neq q$, as demonstrated by the numerical examples in Sec.~\ref{sec:num}.
Candidate patterns are generated in the modulation domain and mapped to the
corresponding code-field symbols before the parity-check test.
Under the $M=q$ presentation, each field symbol is mapped one-to-one to a
modulation symbol.
Assume a memoryless channel and let
\begin{equation}
\mathbf{y}=(y_0,\ldots,y_{n-1}),
\end{equation}
denote the received vector. The log-likelihood (LL) of modulation symbol
$u_i$ at coordinate $j$ is defined as
\begin{equation}
a_{i,j}\triangleq \log f(y_j| u_i),
\qquad
0\leq i\leq q-1,\quad 0\leq j\leq n-1.
\label{eq:local_likelihood}
\end{equation}
For complex AWGN, $a_{i,j}$ is, up to terms independent of $u_i$,
proportional to the negative squared Euclidean distance between $y_j$
and $u_i$.

\textcolor{black}{The LLs in \eqref{eq:local_likelihood} are arranged in a $q\times n$ array. Fig.~\ref{fig:sorted_array}(a) illustrates the construction for a $q=4$, $n=3$
example.} The decoder first sorts the likelihood array in two steps. For each original coordinate $j$, vertical sorting orders the $q$ local symbol candidates from most likely to least likely. Let $\pi_j(s)$ denote the index of the $s$-th ranked symbol in coordinate $j$, so that
\begin{equation}
    a_{\pi_j(0),j}
    \geq a_{\pi_j(1),j}
    \geq \cdots
    \geq a_{\pi_j(q-1),j} .
    \label{eq:vertical_raw}
\end{equation}
\textcolor{black}{In Fig.~\ref{fig:sorted_array}(a), $\pi_0(\cdot )$ sorts column $0$, so that $\pi_0(0)=0$, $\pi_0(1)=2$, $\pi_0(2)=3$ and $\pi_0(3)=1$. Functions $\pi_1(\cdot )$ and $\pi_2(\cdot )$ sort columns $1$ and $2$, respectively. Fig.~\ref{fig:sorted_array}(b) shows the vertically sorted pair array. A pair $(s,\text{LL})$ in row $i$ and column $j$ means that modulation-symbol index $s$ in coordinate $j$ has vertical rank $i$ and log-likelihood $\text{LL}$. Horizontal sorting then orders the sorted coordinates (the sorted columns) by their first-deviation costs. Fig.~\ref{fig:sorted_array}(c) is the horizontally sorted triplet array of Fig.~\ref{fig:sorted_array}(b). A triplet $(s,c,\text{LL})$ in  row $i$ and column $j$ means that modulation-symbol index $s$ in original coordinate $c$ has log-likelihood $\text{LL}$ and its vertical and horizontal ranks are $i$ and $j$, respectively.  The sorted column $2$ of Fig.~\ref{fig:sorted_array}(b) comes first in Fig.~\ref{fig:sorted_array}(c), followed by column $1$ and then column $0$.} Let $h(j)$ be the original coordinate (column) placed in ordered column $j$. The ordering is chosen such that
\begin{equation}
	\begin{aligned}
		&a_{\pi_{h(j)}(0),h(j)}-a_{\pi_{h(j)}(1),h(j)} \\
		&\qquad\leq
		a_{\pi_{h(j+1)}(0),h(j+1)}
		-a_{\pi_{h(j+1)}(1),h(j+1)},
	\end{aligned}
	\label{eq:horizontal_raw}
\end{equation}
for $0\leq j<n-1$. Define the sorted likelihoods
\begin{equation}
	\alpha_j(s)\triangleq a_{\pi_{h(j)}(s),h(j)},
	\qquad 0\leq s\leq q-1 .
	\label{eq:alpha_def}
\end{equation}
Then
\begin{equation}
	\alpha_j(0)\geq \alpha_j(1)\geq\cdots\geq \alpha_j(q-1),
	\label{eq:vert}
\end{equation}
and
\begin{equation}
	\alpha_j(0)-\alpha_j(1)
	\leq
	\alpha_{j+1}(0)-\alpha_{j+1}(1),
	\qquad 0\leq j<n-1 .
	\label{eq:horz}
\end{equation}

\textcolor{black}{In Fig.~\ref{fig:sorted_array}, the ordered coordinates become $h=[2,1,0]$ because the
first-deviation costs in Fig.~\ref{fig:sorted_array}(b) are ranked as
$(13-13)<(12-10)<(16-12)$.}

Property \eqref{eq:vert} says that moving downward within a coordinate cannot
increase likelihood, while \eqref{eq:horz} orders the first-step deviation costs
from left to right. 
\textcolor{black}{Using the fully sorted array, a graph with a $q\times n$ node matrix is constructed. Each node corresponds to a triplet in the sorted array.} A candidate symbol-deviation pattern is represented by a row-index vector
\begin{equation}
    \mathbf r=(r_0,\ldots,r_{n-1})
    \in\{0,\ldots,q-1\}^n .
    \label{eq:row_index_vector}
\end{equation}
The entry $r_j$ means that ordered coordinate $j$ uses its $r_j$-th ranked
local symbol. \textcolor{black}{A pattern is represented in the graph by a path that is created by selecting one node in each coordinate and connecting the selected nodes.} The corresponding code-symbol vector is denoted by
$\mathbf p(\mathbf r)$, and its log-likelihood is
\begin{equation}
    \mathrm{LL}(\mathbf r)=\sum_{j=0}^{n-1}\alpha_j(r_j).
    \label{eq:ll}
\end{equation}
The root pattern is $\mathbf{r}^{(0)}=(0,\ldots,0)$, which selects the most likely symbol in every coordinate. \textcolor{black}{In Fig.~\ref{fig:sorted_array}, these symbols are $(\pi_2(0),\pi_1(0),\pi_0(0))=(3,2,0)$ in the sorted array, corresponding to $\mathbf{p}(\mathbf{r}^{(0)})=(0,2,3)$ with $\mathrm{LL}=13+12+16$.}

\begin{figure}[t]
\centering
\small
\setlength{\tabcolsep}{2pt}
\renewcommand{\arraystretch}{1.05}
\begin{tabular}{@{}c@{\hspace{1mm}}c@{\hspace{1mm}}c@{}}
{\scriptsize
\begin{tabular}{cc|ccc}
\multicolumn{2}{c}{} & \multicolumn{3}{c}{$j$}\\
\multicolumn{2}{c}{} & $0$ & $1$ & $2$\\ \cline{3-5}
    & $0$ & $16$ & $10$ & $5$\\
 \raisebox{-1.0ex}{$i$}    & $1$ & $4$ & $8$ & $12$\\
 & $2$ & $12$ & $12$ & $13$\\
    & $3$ & $11$ & $9$ & $13$
\end{tabular}
}
&
{\scriptsize
\begin{tabular}{c|ccc}
	\multicolumn{1}{c}{} & \multicolumn{3}{c}{$j$}\\
	\multicolumn{1}{c}{} & $0$ & $1$ & $2$\\ \cline{2-4}
	$0$ & $(0,16)$ & $(2,12)$ & $(3,13)$\\
	$1$ & $(2,12)$ & $(0,10)$ & $(2,13)$\\
	$2$ & $(3,11)$ & $(3,9)$ & $(1,12)$\\
	$3$ & $(1,4)$ & $(1,8)$ & $(0,5)$
\end{tabular}}
&
{\scriptsize
\begin{tabular}{c|ccc}
\multicolumn{1}{c}{} & \multicolumn{3}{c}{$j$}\\
\multicolumn{1}{c}{} & $0$ & $1$ & $2$\\ \cline{2-4}
$0$ & $(3,2,13)$ & $(2,1,12)$ & $(0,0,16)$\\
$1$ & $(2,2,13)$ & $(0,1,10)$ & $(2,0,12)$\\
$2$ & $(1,2,12)$ & $(3,1,9)$ & $(3,0,11)$\\
$3$ & $(0,2,5)$ & $(1,1,8)$ & $(1,0,4)$
\end{tabular}
}\\
{\footnotesize (a) $a_{i,j}$}
&
{\footnotesize (b) $(\pi_j(i), a_{\pi_j(i),j})$}
&
{\footnotesize (c)  $(\pi_{h(j)}(i), h(j), \alpha_j(i))$}
\end{tabular}
\vspace{-2mm}
\caption{Non-binary likelihoods for a $q=4$, $n=3$ example. (a) Original array. (b) Vertically sorted array. (c) Vertically and horizontally sorted array.}
\label{fig:sorted_array}
\end{figure}

For every non-root pattern, define its canonical generator by
\begin{equation}
    j^\ast\triangleq \max\{j:r_j>0\},
    \qquad
    i^\ast=r_{j^\ast} .
    \label{eq:generator_def}
\end{equation}
For the root pattern, we use the trivial generator $(i^\ast,j^\ast)=(-1,-1)$. Thus, $j^\ast$ is the rightmost active coordinate and $i^\ast$ is the current
row used in that coordinate. From a pattern with generator $(i^\ast,j^\ast)$,
$q$-ary GRAND creates up to three children:
\begin{itemize}
\item \emph{First child:} if $j^\ast<n-1$, start a new deviation by setting
$r_{j^\ast+1}=1$.
\item \emph{Second child:} if $i^\ast<q-1$ and $j^\ast \geq 0$, deepen the current deviation by
replacing $r_{j^\ast}=i^\ast$ with $r_{j^\ast}=i^\ast+1$.
\item \emph{Third child:} if $i^\ast=1$ and $j^\ast<n-1$, shift a depth-one
deviation to the next ordered coordinate by setting $r_{j^\ast}=0$ and
$r_{j^\ast+1}=1$.
\end{itemize}
The first child starts a new deviation, the second child deepens the current
deviation, and the third child performs the exchange move enabled
by \eqref{eq:vert} and \eqref{eq:horz}. Cf. Fig.~\ref{fig:gen_child}.
\begin{figure}[t!]
    \begin{subfigure}[t]{0.2\textwidth}
\vspace{-0.5cm}
\setlength{\unitlength}{0.65mm}
\begin{picture}(60, 40)(15,0)
\put(0,10){\circle*{2}}
\put(20,25){\circle*{2}}
\put(10,0){\circle*{2}}

\put(0,19){\scriptsize row $i^\ast$}
\put(0,14){\scriptsize row $i^\ast\! +\! 1$}
\put(20,20){\circle{2}}
\put(20,15){\circle{2}}
\put(20,19){\vector(0,-1){3}}
\put(40,29){\vector(0,-1){3}}

\put(18,36){\scriptsize col.}
\put(19,32){\scriptsize $j^\ast$}

\put(0,29){\scriptsize row $0$}
\put(0,24){\scriptsize row $1$}
\put(20,30){\circle*{2}}

\put(40,25){\circle{2}}

\put(38,36){\scriptsize col.}
\put(35,32){\scriptsize $j^\ast \! +\! 1$}
\put(40,30){\circle*{2}}
\put(40,20){\circle*{2}}
\put(40,15){\circle*{2}}
\put(57,36){\scriptsize col.}
\put(55,32){\scriptsize $j^\ast\! + \! 2$}
\put(60,30){\circle*{2}}
\put(60,25){\circle*{2}}
\put(60,20){\circle*{2}}
\put(60,15){\circle*{2}}
\put(40,25){\line(4,1){20}}

\put(11,9){$\mathbf{r}$}
\put(18,9){\scriptsize 2nd child}

\put(20,20){\line(4,1){20}}

\put(10,0){\line(2,3){10}}
\put(20,15){\line(4,3){20}}
\thicklines
\put(0,10){\line(1,-1){10}}

\put(10,0){\line(1,2){10}}

\put(44,23){\scriptsize 1st child}

\put(20,20){\line(2,1){20}}
\put(40,30){\line(1,0){20}}

\end{picture}
    \end{subfigure}
    \begin{subfigure}[t]{0.15\textwidth}
        
\vspace{-0.5cm}
\setlength{\unitlength}{0.6mm}
\begin{picture}(30,40)(-5,0)
\put(0,10){\circle*{2}}
\put(10,0){\circle*{2}}
\put(10,0){\line(1,3){10}}

\put(0,19){\scriptsize row $1$}
\put(20,20){\circle{2}}

\put(15,36){\scriptsize col.}
\put(16,32){\scriptsize $j^\ast$}

\put(0,29){\scriptsize row $0$}
\put(20,30){\circle*{2}}

\put(40,20){\circle{2}}

\put(38,36){\scriptsize col.}
\put(35,32){\scriptsize $j^\ast\! +\! 1$}

\put(40,30){\circle*{2}}
\put(20,30){\line(2,-1){20}}

\put(50,30){\circle*{2}}
\put(60,30){\circle*{2}}
\put(50,20){\circle*{2}}
\put(60,20){\circle*{2}}

\put(40,20){\line(1,1){10}}

\put(17,9){$\mathbf{r}$}

\thicklines
\put(0,10){\line(1,-1){10}}
\put(20,22){\vector(0,1){6}}
\put(40,28){\vector(0,-1){6}}
\put(10,0){\line(1,2){10}}
\put(20,20){\line(2,1){20}}
\put(40,30){\line(1,0){10}}
\put(50,30){\line(1,0){10}}
\put(46,23){\scriptsize 3rd child}
\end{picture}
    \end{subfigure}%
    \centering
    \vspace{-1mm}
    \caption{Three child-generation moves from a parent row-index pattern $\mathbf r$: the 1st and 2nd children (left) and the 3rd exchange child (right). Hollow circles mark the generator nodes.}

    \label{fig:gen_child}
\end{figure}
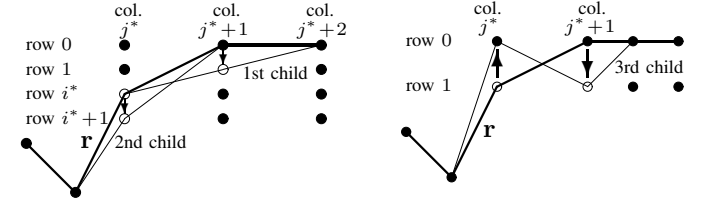

The vector $\mathbf r$ is mapped back to $\mathbf p(\mathbf r)$ by
assigning to original coordinate $h(j)$ the field symbol corresponding to the modulation-symbol index $\pi_{h(j)}(r_j)$ associated with the
$r_j$-th ranked candidate in ordered coordinate $j$. The candidate is valid if
\begin{equation}
    \mathbf{H}\mathbf p(\mathbf r)^T=\mathbf 0 .
    \label{eq:parity_check}
\end{equation}
Equivalently, each candidate can be represented as a deviation from the most-likely word, allowing the code-membership test to be implemented in the syndrome domain. Related connections between guessing noise patterns and guessing codewords are discussed in \cite{gcd}.
The complete decoding procedure is summarized in Algorithm~\ref{alg:qarygrand}.
The decoder maintains a priority queue ordered by \eqref{eq:ll}. At each step
it extracts the currently most likely pattern, checks whether it satisfies
\eqref{eq:parity_check}, and otherwise inserts up to three children according
to the rules above.

Section~\ref{sec:ml} uses exactly the two ordering properties in
\eqref{eq:vert} and \eqref{eq:horz} to show that the search induced by
Algorithm~\ref{alg:qarygrand} is globally likelihood-ordered.

\section{ML Optimality}\label{sec:ml}
We now prove that the local child-generation rules of Section~\ref{sec:qary}
recover the global likelihood order on the full row-index space. Assume
uniform codeword priors, a memoryless channel, a one-to-one symbol mapping, and
deterministic tie-breaking. The proof has two ingredients: likelihood monotonicity along every generated edge, and
the fact that the generated graph is a rooted spanning tree.

\vspace{-2mm}

\begin{lemma}[Edge monotonicity]\label{lemma:em}
Every child generated by $q$-ary GRAND has LL no larger than
that of its parent.
\end{lemma}

\begin{proof}
Let $\mathbf r$ be a parent pattern with generator $(i^\ast,j^\ast)$. For
the first child, only coordinate $j^\ast+1$ changes, from row $0$ to row $1$.
Hence
\begin{equation}
\mathrm{LL}(\mathbf r_{\rm first})-\mathrm{LL}(\mathbf r)
=
\alpha_{j^\ast+1}(1)-\alpha_{j^\ast+1}(0)\leq0,
\label{eq:first_mono}
\end{equation}
by \eqref{eq:vert}. For the second child, only coordinate $j^\ast$ changes,
from row $i^\ast$ to row $i^\ast+1$, and therefore
\begin{equation}
\mathrm{LL}(\mathbf r_{\rm second})-\mathrm{LL}(\mathbf r)
=
\alpha_{j^\ast}(i^\ast+1)-\alpha_{j^\ast}(i^\ast)\leq0.
\label{eq:second_mono}
\end{equation}
For the third child, a depth-one deviation is moved from $j^\ast$ to
$j^\ast+1$. Thus
\begin{equation}
\begin{aligned}
\mathrm{LL}(\mathbf r_{\rm third})-\mathrm{LL}(\mathbf r)
&=\bigl[\alpha_{j^\ast}(0)-\alpha_{j^\ast}(1)\bigr] \\
&\quad-\bigl[\alpha_{j^\ast+1}(0)-\alpha_{j^\ast+1}(1)\bigr]\leq0,
\end{aligned}
\label{eq:third_mono}
\end{equation}
where the last inequality follows from \eqref{eq:horz}. Hence all generated
edges are likelihood non-increasing.
\end{proof}

\begin{algorithm}[!t]
\caption{$q$-ary GRAND}
\footnotesize
\label{alg:qarygrand}
\begin{algorithmic}[1]
\Require Parity-check matrix $\mathbf H$, received vector $\mathbf y$, budget $B$
\State Construct the sorted likelihood array using \eqref{eq:vert}--\eqref{eq:horz}
\State Initialize root $\mathbf r^{(0)}=(0,\ldots,0)$ and queue $S\gets\{(-1,-1,\mathbf r^{(0)})\}$
\For{$m=1$ to $B$}
  \State Extract $(i^\ast,j^\ast,\mathbf r)$ of maximum likelihood from $S$
  \If{$\mathbf H\mathbf p(\mathbf r)^T=\mathbf 0$}
    \State \Return $\mathbf p(\mathbf r)$
  \EndIf
  \If{$j^\ast<n-1$}
    \State $\mathbf r'\gets\mathbf r$; $r'_{j^\ast+1}\gets1$; insert first child $(1,j^\ast+1,\mathbf r')$
    \If{$j^\ast\ge 0 \land i^\ast=1$}
      \State $\mathbf r''\gets\mathbf r$; $r''_{j^\ast}\gets0$, $r''_{j^\ast+1}\gets1$; insert third child $(1,j^\ast+1,\mathbf r'')$
    \EndIf
  \EndIf
  \If{$j^\ast\ge 0 \land i^\ast<q-1$}
    \State $\mathbf r'''\gets\mathbf r$; $r'''_{j^\ast}\gets i^\ast+1$; insert second child $(i^\ast+1,j^\ast,\mathbf r''')$
  \EndIf
\EndFor
\State \Return failure
\end{algorithmic}
\end{algorithm}
\vspace{-2.5mm}

\begin{lemma}[Spanning-tree property]\label{lemma:stp}
The child-generation rules induce a rooted spanning tree on
$\{0,\ldots,q-1\}^n$, rooted at $\mathbf r^{(0)}=(0,\ldots,0)$.
\end{lemma}

\begin{proof}
Let $\mathbf e_j$ denote the unit vector in coordinate $j$. For every
non-root pattern $\mathbf r$, using its generator in \eqref{eq:generator_def}, define a parent map $\mathrm{Par}(\mathbf r)$ by
\begin{equation}
\mathrm{Par}(\mathbf r)\triangleq
\begin{cases}
\mathbf r-\mathbf e_{j^\ast}, & i^\ast>1,\\
\mathbf r-\mathbf e_{0}, & i^\ast=1,\ j^\ast=0,\\
\mathbf r-\mathbf e_{j^\ast},
& i^\ast=1,\ j^\ast\geq1,\ r_{j^\ast-1}>0,\\
\mathbf r-\mathbf e_{j^\ast}+\mathbf e_{j^\ast-1},
& i^\ast=1,\ j^\ast\geq1,\ r_{j^\ast-1}=0 .
\end{cases}
\label{eq:parent_map_short}
\end{equation}
The four cases are mutually exclusive. In the first case, $\mathbf r$ is
obtained from $\mathrm{Par}(\mathbf r)$ by a second-child move. In the second case,
$\mathbf r=\mathbf e_0$ is obtained from the root by a first-child move. In
the third case, $\mathbf r$ is obtained from $\mathrm{Par}(\mathbf r)$ by a first-child
move, and in the fourth case by a third-child move. Thus every non-root
pattern has a unique parent under the proposed rules.

It remains to show that every parent chain reaches the root. Define
\begin{equation}
\Phi(\mathbf r) \triangleq \sum_{j=0}^{n-1} r_jq^j.
\label{eq:potential_short}
\end{equation}
In each case of \eqref{eq:parent_map_short},
\begin{equation}
\Phi(\mathrm{Par}(\mathbf r))<\Phi(\mathbf r).
\label{eq:descent_short}
\end{equation}
The first three cases remove a positive contribution at $j^\ast$, while the
fourth replaces $q^{j^\ast}$ by $q^{j^\ast-1}$. Hence cycles are impossible,
and repeated parent application must terminate at the unique vector with
$\Phi=0$, namely the root. Therefore every row-index vector is connected to
the root by a unique parent chain, and the generated graph is a rooted
spanning tree.
\end{proof}
The tree generates each row-index pattern once, so no duplicate-removal table is required.

\begin{theorem}[ML optimality]
Under unlimited search, Algorithm~\ref{alg:qarygrand} returns a
maximum-likelihood codeword.
\end{theorem}

\begin{proof}
By Lemma~\ref{lemma:em}, likelihood is non-increasing along every generated edge. By
Lemma~\ref{lemma:stp}, the generated graph is a rooted spanning tree of the full row-index
space. At any stage of the best-first search, the priority queue contains
the frontier of the explored subtree. Every unseen pattern is a descendant of
some frontier pattern and therefore has likelihood no larger than that
frontier ancestor. Hence the most likely unseen pattern always lies on the
frontier, and best-first extraction tests patterns in global non-increasing
likelihood order.

Let $\mathbf p(\mathbf r^\ast)$ be the first tested candidate satisfying
$\mathbf{H}\mathbf p(\mathbf r^\ast)^T=\mathbf0$. If another codeword
$\mathbf p(\mathbf r)$ had strictly larger likelihood, then $\mathbf r$ would
have appeared earlier in the globally ordered test sequence and would have
been accepted before $\mathbf r^\ast$, a contradiction. Therefore
$\mathbf p(\mathbf r^\ast)$ is maximum-likelihood.
\end{proof}

\section{Algorithmic Cost and Implementation Considerations}\label{sec:wc}
\label{sec:complexity}
The computational cost of $q$-ary GRAND is determined by more than the
number of tested patterns. Let $m$ denote the number of candidate patterns
submitted to the parity-check test before termination, and let $F$ denote
the peak number of live patterns in the search frontier. We distinguish
preprocessing, finite-field arithmetic, frontier management, storage, and
query count. The quantity $m$ is retained as a measure of search depth
rather than as a complete measure of implementation complexity. With a
finite budget $B$, one has $m\leq B$, while under unlimited search $m$ is
the stopping time at which the first valid codeword is found.
The one-time preprocessing consists of computing and sorting the symbol
likelihoods. Computing all local metrics costs $O(nq)$, vertical sorting
costs $O(nq\log q)$, and horizontal sorting of the $n$ first-deviation
costs requires $O(n\log n)$. The total preprocessing cost is therefore
\begin{equation}
O(nq\log q+n\log n).
\end{equation}

\begin{figure}[!t]
\vspace{-8mm}
\centering
\includegraphics[width=\columnwidth]{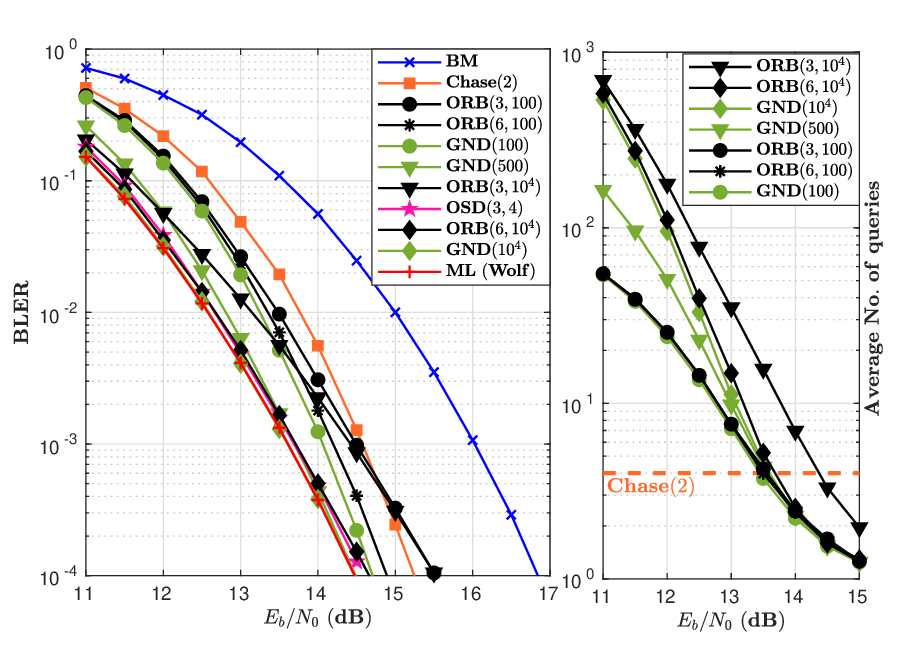}\vspace{-.3cm}
\caption{Performance of decoders for RS$(20,18)$ over $\mathbb F_{64}$ with $64$-QAM. BLER is shown on the left and the average number of queries on the right.}
\label{fig:rs20_18}
\vspace{-.3cm}
\end{figure}

During the search, each expansion requires one priority-queue extraction
and generates at most three children. Let $\rho=n-k$ denote the code redundancy. With incremental syndromes, a child that changes one coordinate requires $O(\rho)$ finite-field operations, while
the exchange child changes two coordinates and has the same asymptotic
cost up to a constant factor. Priority-queue insertion and extraction
require $O(\log F)$ operations. The resulting post-preprocessing
complexity for $m$ tested patterns is therefore
\begin{equation}
O\!\left(m(\rho+\log F)\right).
\end{equation}
Since the search graph is a tree, each row-index pattern has a unique
parent and no duplicate-pattern table is required.
Exact best-first ordering nevertheless requires a dynamic frontier. If a
direct implementation stores the full row-index vector and syndrome for
every live node, the storage scales as
\begin{equation}
O\!\left(F(n+\rho)\right).
\end{equation}
Parent-pointer or difference-based representations can reduce the
pattern-description storage, at the expense of reconstruction or
additional cached state. Thus, the bounded three-child expansion does not
imply bounded overall memory, since $F$ remains dependent on the received
vector and search depth. This dynamic-frontier requirement is also
characteristic of exact soft-input SGRAND~\cite{sgrand}. A distinguishing feature of the proposed construction is its canonical three-child, unique-parent
organization, which preserves exact likelihood ordering without duplicate
suppression.
Symbol-level ORBGRAND~\cite{sorb} makes a different implementation
tradeoff. It replaces exact likelihood ordering by a more regular
approximate schedule that is attractive for parallel hardware. Its
scheduler and storage can consequently be provisioned more regularly than
the received-vector-dependent frontier of exact best-first search. Its
generated patterns may include invalid symbol combinations that are
discarded before the parity-check test, so raw generated patterns and
parity-check queries are distinct quantities. In contrast, the proposed
tree generates row-index patterns directly, while retaining the
priority-queue and dynamic-memory costs associated with exact best-first
ordering. OSD~\cite{foss} has a different complexity structure again,
requiring reliability ordering and code-dependent most-reliable-basis
reprocessing followed by enumeration of a predetermined set of test
codewords. These distinctions prevent query count alone from establishing
universal implementation superiority. The proposed decoder provides exact
likelihood ordering, bounded local generation, and data-dependent
stopping, whereas ORBGRAND offers a more regular generation structure
suited to highly parallel implementations and OSD exchanges dynamic search
for code-dependent reprocessing.

\begin{figure}[!t]
\vspace{-8mm}
\centering
\includegraphics[width=\columnwidth]{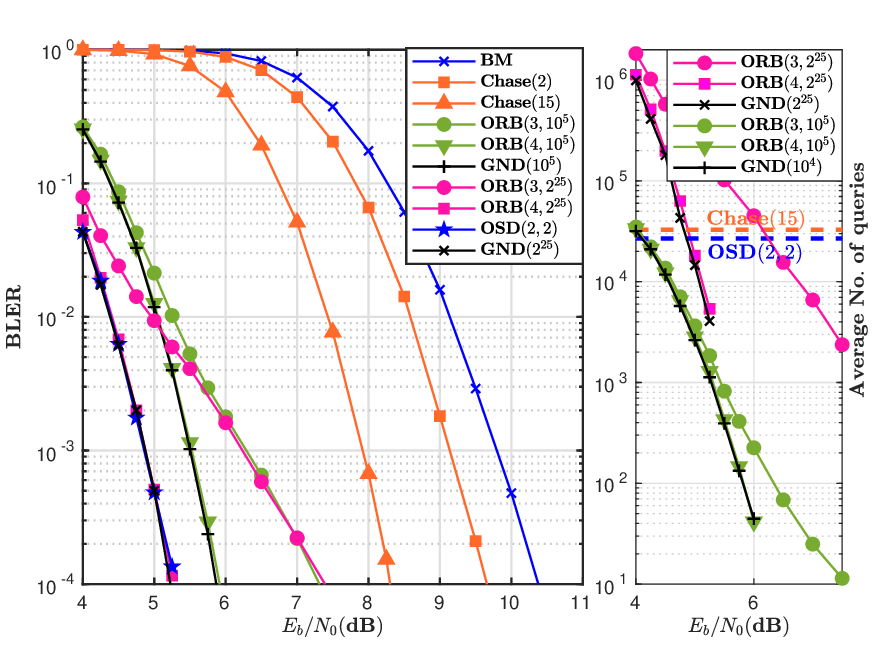}\vspace{-.3cm}
\caption{Same as Fig. \ref{fig:rs20_18} but for eBCH$(256,231)$ with QPSK.}
\label{fig:bch256_231}
\vspace{-.3cm}
\end{figure}

\section{Numerical Results}\label{sec:num}
We consider two simulation settings. The first evaluates $q$-ary GRAND on RS$(20,18)$ over $\mathbb F_{64}$ with 64-QAM. The code has rate $R=0.9$, minimum distance $d=3$, and throughput $\eta=R\log_2(64)=5.4$ bits per channel use (bpcu). The second considers a longer, high-rate binary QPSK-modulated extended BCH$(256,231)$ code with code rate $R=231/256 \approx 0.9$, minimum distance $d=8$, and throughput $\eta=R\log_2(4) \approx 1.8$ bpcu. In this case, the modulation alphabet size is $M=4$, whereas the code alphabet size is $q=2$. Test patterns are generated in the modulation domain, while the parity check is performed after mapping each pattern to the binary code domain. The RS and BCH codes are implemented in nonsystematic form using Fourier transform (FT) over finite-field constructions \cite{bla-a}. The RS code is constructed using $21$-point FT followed by shortening a symbol. The eBCH$(256,231)$ code is obtained by extending the corresponding length-$255$ BCH code using $255$-point FT. All constellations are normalized to unit average symbol energy, and the channel is complex AWGN with distribution $\mathcal{CN}(0;N_0)$. We define signal-to-noise ratio (SNR) as $E_s/N_0=R\log_2(M)E_b/N_0$, where $E_s$ is the average modulation-symbol energy and $E_b$ is the information-bit energy. A finite budget $B$ means that at most $B$ row-index patterns are tested.

We compare $q$-ary GRAND with several benchmark decoders. BM decoding~\cite{bla-a} is used as a hard-decision baseline. Chase decoding~\cite{chase} provides a small-list soft-decision reference. It selects the $L$ least reliable symbols, tests $2^L$ hard-decision patterns using BM decoding, and retains the valid codeword with minimum Euclidean distance to the received vector. Symbol-level ORBGRAND~\cite{sorb} considers the $\mu$ most likely options for each received symbol, sorts their reliabilities, approximates the ordered reliabilities by a three-segment linear function, and generates test patterns according to this approximation. Invalid patterns are discarded. We also evaluate ordered-statistics decoding (OSD)~\cite{foss}. To keep its complexity manageable, OSD generates test codewords by modifying up to $t$ symbols among the $k$ most reliable symbols, with the $\mu$ most likely symbol options considered at each selected position. For the RS$(20,18)$ case, Wolf's syndrome-trellis decoder~\cite{wolf} is used as the exact ML reference.

Fig.~\ref{fig:rs20_18} reports the results for RS$(20,18)$ over $\mathbb F_{64}$ with 64-QAM. Wolf's syndrome-trellis decoder is feasible for this high-rate code because its complexity is governed primarily by the redundancy $n-k$. For OSD, test codewords are generated by modifying up to $t=4$ symbols among the $18$ most reliable symbols, with $\mu=3$ most likely symbol options considered at each selected position. Hence, OSD tests \(N_{\mathrm{OSD}}=\sum_{w=0}^{4}\binom{18}{w}2^w =56{,}137 \) codewords per received block. As shown in Fig.~\ref{fig:rs20_18}, $q$-ary GRAND with $B=10^4$ operates very close to the ML benchmark. OSD is also close to ML, as is symbol-level ORBGRAND with $\mu=6$ and $B=10^4$ valid test patterns.

Fig.~\ref{fig:bch256_231} reports the results for the QPSK-modulated binary extended BCH$(256,231)$ code. Symbol-level ORBGRAND with $\mu=4$ and $q$-ary GRAND achieve similar performance for a budget of $B=10^5$, whereas the performance of ORBGRAND with $\mu=3$ degrades at high $E_b/N_0$. For OSD, bit log-likelihoods are first calculated from the $128$ received QPSK symbols, and decoding is then performed in the binary domain. With a maximum test-pattern weight of $t=2$, OSD tests \( N_{\mathrm{OSD}} =\sum_{w=0}^{2}\binom{231}{w}
=26{,}797
\) codewords per received block. In this case, OSD achieves performance similar to the GRAND decoders with $B=2^{25}=33{,}554{,}432$. BM and Chase decoding with $L=2$ and $L=15$ exhibit a substantial SNR loss compared with GRAND and OSD. GRAND-based algorithms achieve near-ML performance with a budget of $B=q^{n-k}$ queries. For a fixed $k$, as $n$ increases (i.e., code rate decreases), the performance of the algorithms degrades when the computation burden limits $B$ to a level much lower than $q^{n-k}$.

In the figures, ORBGRAND, Chase, OSD, and $q$-ary GRAND are denoted by $\mathbf{ORB}(\mu,B)$, $\mathbf{Chase}(L)$, $\mathbf{OSD}(\mu, t)$ and $\mathbf{GND}(B)$, respectively.

The number of tested patterns is the main finite-budget search statistic. Unlike Chase decoding and OSD, which test a fixed number of candidate patterns or codewords, $q$-ary GRAND has a data-dependent stopping time. Figs.~\ref{fig:rs20_18} and~\ref{fig:bch256_231} show that the average number of tested patterns decreases rapidly with SNR. For RS$(20,18)$, the average count is approximately $2.2$ at $14$ dB and $1.5$ at $14.5$ dB, already below the Chase list size of four. Thus, the finite-budget advantage is not limited to BLER performance. The number of tested patterns also becomes small in the same SNR region where the decoder approaches ML performance. The average valid-query curves for ORBGRAND with $\mu=6$ for $M=64$ and with $\mu=4$ for $M=4$ are close to those of the corresponding $q$-ary GRAND cases. Overall, unlimited-search $q$-ary GRAND reproduces exact ML decoding in the representative non-binary RS benchmark, while finite-budget $q$-ary GRAND provides an advantageous search depth on average at high SNR values.

\section{Conclusion}\label{sec:con}
We introduced $q$-ary GRAND as a likelihood-ordered symbol-level decoder for binary and non-binary linear codes with higher-order modulation. Its three local generation rules form a monotone spanning tree over the symbol-deviation space and therefore recover ML decoding under unlimited search. Agreement with Wolf decoding in the RS case verifies the ML result numerically. The RS and eBCH experiments further show that finite search budgets retain most of this decoding gain while the average number of tested patterns falls to only a few at high SNR.
\balance
\bibliography{qgrand} \bibliographystyle{IEEEtran} 

@article{chase,
  title={Class of algorithms for decoding block codes with channel measurement information},
  author={Chase, David},
  journal={IEEE Trans. Inf. Theory},
  volume={18},
  number={1},
  pages={170--182},
  year={1972}
  }

@article{wolf,
  title={Efficient maximum likelihood decoding of linear block codes using a trellis},
  author={Wolf, Jack},
  journal={IEEE Trans. Inf. Theory},
  volume={24},
  number={1},
  pages={76--80},
  year={1978},
  publisher={IEEE}
}

@article{duffy2019grand,
  title={Capacity-achieving guessing random additive noise decoding},
  author={Duffy, Ken R and Li, Jiange and M{\'e}dard, Muriel},
  journal={IEEE Trans. Inf. Theory},
  volume={65},
  number={7},
  pages={4023--4040},
  year={2019},
  publisher={IEEE}
}

@inproceedings{grand_markov,
  title={Guessing noise, not code-words},
  author={Duffy, Ken R and Li, Jiange and M{\'e}dard, Muriel},
  booktitle={2018 IEEE Int. Symp. Inf. Theory (ISIT)},
  year={2018}
}

@article{orbgrand,
  title={Ordered reliability bits guessing random additive noise decoding},
  author={Duffy, Ken R and An, Wei and M{\'e}dard, Muriel},
  journal={IEEE Trans. Signal Process.},
  volume={70},
  pages={4528--4542},
  year={2022},
  publisher={IEEE}
}

@inproceedings{sorb,
  title={Soft decoding without soft demapping with {ORBGRAND}},
  author={An, Wei and M{\'e}dard, Muriel and Duffy, Ken R},
  booktitle={IEEE Int. Symp. Inf. Theory (ISIT)},
  pages={1080--1084},
  year={2023}
}

@inproceedings{orbcc,
  title={Using channel correlation to improve decoding-{ORBGRAND-AI}},
  author={Duffy, Ken R and Grundei, Moritz and M{\'e}dard, Muriel},
  booktitle={IEEE Global Commun. Conf. (GLOBECOM)},
  year={2023}
}

@inproceedings{orb-hw,
  title={Low-latency modulation- and correlation-adaptive {ORBGRAND-AI} decoder},
  author={Kizilates, Zeynep Ece and Riaz, Arslan and Bali, Akshaya and Grundei, Moritz and M{\'e}dard, Muriel and Duffy, Ken R and Yazicigil, Rabia Tugce},
  booktitle={2025 IEEE Eur. Solid-State Electron. Res. Conf. (ESSERC)},
  year={2025}
}

@article{sorb-fc,
  title={Symbol-level {GRAND} for high-order modulation over block fading channels},
  author={Chatzigeorgiou, Ioannis and Monteiro, Francisco A},
  journal={IEEE Commun. Lett.},
  volume={27},
  number={2},
  pages={447--451},
  year={2023}
}

@article{orbRS,
  title={Group Probability Decoding of Turbo Product Codes over Higher-Order Fields},
  author={Rapp, Lukas and M{\'e}dard, Muriel and Duffy, Ken R},
  journal={IEEE Trans. Commun.},
  year={2026}
}

@article{sog,
  title={Soft-output ({SO}) {GRAND} and iterative decoding to outperform {LDPC} codes},
  author={Yuan, Peihong and M{\'e}dard, Muriel and Galligan, Kevin and Duffy, Ken R},
  journal={IEEE Trans. Wireless Commun.},
  volume={24},
  number={4},
  pages={3386--3399},
  year={2025},
  publisher={IEEE}
}

@inproceedings{sgrand,
  title={Soft maximum likelihood decoding using {GRAND}},
  author={Solomon, Amit and Duffy, Ken R and M{\'e}dard, Muriel},
  booktitle={IEEE Int. Conf. Commun. (ICC)},
  year={2020}
  }

@phdthesis{a-sol,
  title={Noise-centric decoding},
  author={Solomon, Amit},
  year={2021},
  school={Massachusetts Institute of Technology}
}

@inproceedings{gcd,
  title={Guessing what, noise or codeword?},
  author={Ma, Xiao},
  booktitle={IEEE Inf. Theory Workshop (ITW)},
  year={2024}
}

@article{foss,
  title={Soft-decision decoding of linear block codes based on ordered statistics},
  author={Fossorier, Marc PC and Lin, Shu},
  journal={IEEE Trans. Inf. Theory},
  volume={41},
  number={5},
  pages={1379--1396},
  year={1995}
}

@BOOK{bla-a,
TITLE={Algebraic codes for data transmission},
AUTHOR="Richard~E.~Blahut",
PUBLISHER="Cambridge University Press",
YEAR=2003}
\end{document}